\documentclass{sig-alternate}

\usepackage[utf8]{inputenc}

\usepackage{mathrsfs}

\usepackage{amsthm}
\usepackage{stmaryrd}
 
\usepackage{multirow}

\usepackage[
  citestyle=authoryear-comp,
  bibstyle=authoryear,
  backend=bibtex,
  isbn=false,
  opcittracker=true,
  autocite=inline,
  doi=true,url=false,
  maxcitenames=3,
  maxbibnames=10,
  firstinits=true
]{biblatex}
\bibliography{biblio}

\DeclareNameAlias{sortname}{first-last}

\usepackage{tikz}
\usetikzlibrary{calc, matrix, backgrounds,shapes.geometric}

\usepackage{xspace}

\usepackage{float}
\usepackage[noend]{algpseudocode}
\usepackage[ruled]{caption}      %

\floatstyle{ruled}
\newfloat{algo}{t}{lop}
\floatname{algo}{Algorithm}

\floatstyle{plaintop}
\restylefloat{figure}

\usepackage{xcolor}

\usepackage[pdftex,hypertexnames=false]{hyperref}

\definecolor{darkblue}{rgb}{0,0,0.5}
\definecolor{cerule}{RGB}{53,122,183}
\definecolor{cardinal}{RGB}{184,32,16}

\hypersetup{
    colorlinks,
    linkcolor={black},
    citecolor={cerule},
    urlcolor={cardinal}
}

\usepackage{microtype}

\newtheoremstyle{plain}
  {\topsep}   %
  {\topsep}   %
  {\itshape}  %
  {0pt}       %
  {\scshape} %
  {.}         %
  {5pt plus 1pt minus 1pt} %
  {}          %

\newtheoremstyle{definition}
  {\topsep}   %
  {\topsep}   %
  {}  %
  {0pt}       %
  {\scshape} %
  {.}         %
  {5pt plus 1pt minus 1pt} %
  {}          %

\newtheorem{theo}{Theorem}
\newtheorem{lem}[theo]{Lemma}
\newtheorem{prop}[theo]{Proposition}

\theoremstyle{definition}
\newtheorem*{defn}{Definition}

\theoremstyle{remark}

\newcommand{\noopsort}[1]{}

\newcommand{\Z}{\mathbb{Z}}
\newcommand{\Zp}{\mathbb{Z}_p}
\newcommand{\Q}{\mathbb{Q}}
\newcommand{\Qp}{\Q_p}
\newcommand{\Fp}{\mathbb{F}_p}

\newcommand\cO{\mathcal{O}}

\newcommand\cY{\mathcal{Y}}

\newcommand{\ud}{\mathrm{d}}

\newcommand\Zpt{\mathbb{Z}_p \llbracket t \rrbracket}
\newcommand\Qpt{\Qp\llbracket t \rrbracket}

\newcommand\Opt{\cO_K \llbracket t \rrbracket}
\newcommand\Kpt{K\llbracket t \rrbracket}

\newcommand\padic{$p$-adic\xspace}
\newcommand\abs[1]{\left|#1\right|}

\def\eqdef{\stackrel{\text{def}}{=}}
\def\leq{\leqslant}
\def\geq{\geqslant}
\def\st{\ \middle|\ }

\def\smallint{\begingroup\textstyle \int\endgroup}

\begin{document}

\CopyrightYear{2016}
\setcopyright{acmcopyright}
\conferenceinfo{ISSAC '16,}{July 19--22, 2016, Waterloo, ON, Canada}
\isbn{978-1-4503-4380-0/16/07}\acmPrice{\$15.00}
\doi{http://dx.doi.org/10.1145/2930889.2930912}

\title{On {\ttlit p}-adic Differential Equations\\ with Separation of Variables}

\numberofauthors{2}
\author{\alignauthor
  Pierre Lairez\titlenote{Partially funded by the DFG grant \mbox{BU~1371/2-2}.}\\
       \affaddr{Technische Universität}\\
       \affaddr{Berlin, Germany}
       \email{pierre@lairez.fr}
\alignauthor
Tristan Vaccon\\
       \affaddr{JSPS-Rikkyo University} \\
       \affaddr{Tokyo, Japan}
       \email{vaccon@rikkyo.ac.jp}
}

\maketitle

\begin{abstract} 
Several algorithms in computer algebra involve the computation of a power series solution of a given ordinary differential equation. Over finite fields, the problem is often lifted in an approximate $p$-adic setting to be well-posed. This raises precision concerns: how much precision do we need on the input to compute the output accurately? In the case of ordinary differential equations with separation of variables, we make use of the recent technique of differential precision to obtain optimal bounds on the stability of the Newton iteration.  The results apply, for example, to algorithms for manipulating algebraic numbers over finite fields, for computing isogenies between elliptic curves or for deterministically finding roots of polynomials in finite fields.  The new bounds lead to significant speedups in practice.
\end{abstract}

\ccsdesc[500]{Computing methodologies~Algebraic algorithms}
\ccsdesc[300]{Computing methodologies~Number theory algorithms}
\ccsdesc[500]{Mathematics of computing~Ordinary differential equations}

\printccsdesc

\keywords{Ordinary differential equation; $p$-adic numbers; Newton iteration; numerical stability; differential precision}

\section{Introduction}
\label{section:presentation-equa-diff}

We study, in a~$p$-adic context, the loss of precision occuring during the
computation of a power series solution of a certain class of differential
equations. We use the method of differential precision that relies on a
first-order analysis.

\subsection{The {\subsecit p}-adic context}

Let~$\Zp$ be the ring of~$p$-adic integer, for a given prime~$p$, and~$\Qp$ its
field of fractions. The~\padic valuation on~$\Qp$ is denoted~$v_p$, and
the~\padic norm is defined by~$|a| = p^{-v_p(a)}$, for~$a\in\Qp$.
For~$a\in\Qp$ and~$\varepsilon$ a positive real number,
let~$a+O(\varepsilon)$ denote the set of all~$b\in\Qp$ such that~$|a-b|\leq
\varepsilon$. %
For~$a,b\in\Zp$, we have~$a = b + O(\varepsilon)$ if and only if~$a\equiv b
\pmod{p^{-\lfloor \log_p \varepsilon \rfloor}}$.

A computer can handle \padic numbers given with bounded precision: a \padic
number~$a$ is approximately represented by a rational number~$a'$ and a
radius~$\varepsilon$ such that~$a \in a' + O(\varepsilon)$. This leads to a 
\emph{ball arithmetic} over~$\Qp$.
The ultrametric nature of~$\Qp$ makes ball arithmetic particularly convenient
since errors do not propagate when adding two numbers:
\[ \left(a+O(\varepsilon)\right) + \left( b + O(\eta) \right) = a + b + O(\max(\varepsilon, \eta)), \]
or multiplying them:
\[  \left(a+O(\varepsilon)\right) \left( b + O(\eta) \right) = ab + O(\max(\eta\abs a, \varepsilon\abs b)), \]
or dividing them:
\[ \frac{a+O(\varepsilon)}{b+O(\eta)} = \frac{a}{b} + O(\max(\eta \abs a \abs{b}^{-2}, \varepsilon \abs{b}^{-1})). \]
These formulae are optimal: the equalities are set equalities and not only
left-to-right inclusions.  When considering an algorithm performing additions,
multiplications and divisions over~\padic numbers, it is possible to track the
precision during all the intermediate steps. Thus, it is possible to run the
algorithm on inputs given approximately as balls, and to return the result as a
ball with the guarantee that whatever the \emph{exact} values of the input are,
the \emph{exact} result lies in the ball returned.  However, even if for every 
single operation the formulae above give the optimal precision of the result,
the optimality does not compose.  This is the well-known \emph{dependency
problem}. It is a major obstacle to the application of ball arithmetic
over~$\Qp$, in the same way as it constricts interval arithmetic over~$\mathbb R$.

For example, let us consider the computation of the determinant of a matrix
with $p$-adic integer coefficients, given at precision~$\varepsilon$.  Since
the determinant is an integral polynomial function of the coefficients, the
determinant is also known at precision at least~$\varepsilon$.  However, if it
is computed through a Gaussian elimination, and if at some point of the
computation, one of the pivots has a positive valuation, then basic precision
tracking will indicate that the result is only correct at precision less
than~$\varepsilon$.  The \emph{intrinsic} loss of precision is null, or even
negative \parencite{CarRoeVac15}, while the \emph{algorithmic} loss, that
depend on the algorithm, may be positive.

\Textcite{CarRoeVac14} have shown, in a $p$-adic setting, how to use first
order analysis to obtain rigorous and optimal precision bounds on both the intrinsic and the algorithmic loss of precision.
The method relies on
the following fact:

\begin{lem}[{\Cite{CarRoeVac14}}] \label{lem:prec}
Let $\varphi : \Qp^n \rightarrow \Qp^m$ be a differentiable
map and let~$x\in\Qp^n$. If the differential~$\ud_x\varphi$ at~$x$
is surjective, then
$ \varphi(x+B)=\varphi(x)+\ud_x \varphi(B)$ for any zero-centered small enough ball $B$.
\end{lem}

In other words, if the precision on the input~$x$ is~$B$, then the precision on
the output~$\varphi(x)$ is~$\ud_x \varphi(B)$; this is the inclusion
of~$\varphi(x+B)$ in~$\varphi(x)+\ud_x \varphi(B)$.  And conversely, every
pertubation of~$\varphi(x)$ up to~$\ud_x \varphi(B)$ comes from a
pertubation from the input~$x$ up to~$B$; this is the converse inclusion
of~$\varphi(x)+\ud_x \varphi(B)$ in~$\varphi(x+B)$. 
What \emph{small enough} is can be
expressed in terms of the norms of the higher differentials of~$\varphi$
at~$x$, in the case that $\varphi$ is analytic at $x$.

\subsection{Main result}

We study here the computation of a power series with $p$-adic coefficients
solution of a given first-order ordinary differential equation with separation of
variables.  Let~$g$ and~$h$ be power series in~$\Zpt$ such
that~$h(0)=1$ and~$g(0)\neq0$. We consider the following differential equation:
\begin{equation} \label{MainDiffEqu} \tag{E}
  y' =g \cdot h(y), \quad y(0) = 0.
\end{equation}
It has a unique solution~$y\in \Qpt$. We make the assumption that this solution
has integer coefficients, that is~$y\in \Zpt$.
Given the~$n$ first coefficients of~$g$ and~$h$, at bounded precision,
how far can we approximate~$y$? At which computational cost?

The case of a general initial condition~$y(0)=c$ reduces to the equation above
by changing~$h(t)$ in~$h(t+c)$.  In particular, the linear equation~$y' =
g\cdot y$, with~$y(0)=1$ can be written~$z' = g\cdot(1+z)$, with~$z(0)=0$,
thanks to the transformation~$z = y-1$.  Nevertheless, the initial condition~$y(0)=0$
ensures that the composition~$h(y)$ is well defined when~$h$ is a general power
series.

\begin{defn}
For~$\varepsilon$ a positive real number and~$n$ a positive integer, an
\emph{approximation modulo~$(p^\kappa, t^n)$} of a power series~$f = \sum_{i\geq0} a_i t^i \in \Zpt$
is a power series~$\bar f= \sum_{i\geq0} b_i t^i\in \Zpt$ such that~$a_i = b_i \pmod{p^\kappa}$
for all~$i < n$.
\end{defn}

The complexity of computing~$y$ depends on the complexity of the power series
multiplication and the composition~$h(f)$ for a general power series~$f$.
Let~$M_\Z(p^\lambda, n)$ the number of bit operations required to compute the
product of two polynomials of degree~$n$ with coefficients
in~$\Z/p^{\lambda}\Z$.  Let~$C_h(p^\lambda, n)$, the number of bit operations
needed to compute an approximation modulo~$(p^\lambda, t^n)$ of~$h(f)$ given an
approximation modulo~$(p^\lambda, t^n)$ of a power series~$f \in \Zpt$.  We
assume that~$C_h(p^\lambda, 2n) \geq 2 C_h(p^\lambda, n)$.  In the general
case, \textcite{KedUma11} proved the quasi-optimal bound~$C_h(p^\lambda, n) =
O\big( (n \lambda \log p)^{1+o(1)}\big)$.  In practice, $h$ is given as a
procedure that computes the composition~$h(f)$ modulo~$(p^\lambda, t^n)$ for
any~$f \in \Zpt$ given modulo~$(p^\lambda, t^n)$.  This
composition is easy to compute in most applications: $h$ is often a rational function of small degree or a
radical of such a rational function so that~$C_h(p^\lambda, n) =
O(M_\Z(p^\lambda,n))$.  We may regard~$h$ as known with infinite precision, but
the computations depends only on a suitable approximation of~$h$.

Our main result is then the following:
\begin{theo} \label{MainTheorem}
  Let~$n > 0$, $\kappa > 0$ (or~$\kappa > 1$ if~$p=2$) and let $\lambda = \kappa + \lfloor \log_p n \rfloor$.
  One can
  compute an approximation modulo $(p^{\kappa}, t^{n+1})$ of the solution~$y$ of~\eqref{MainDiffEqu}
  given approximations modulo~$(p^{\lambda}, t^{n})$ of~$g$ and~$h$,
  using $O\left(M_\Z(p^\lambda,n) + C_h(p^\lambda, n) \right)$ bit operations.
\end{theo}

This result was already known in the linear case: \textcite{BosGonPer05} gave
the first proof and then \textcite{Grenet:2015} gave a simpler one.  In the
non-linear case, \textcite{LerSir08} obtained a weaker bound: they showed that
an approximation modulo~$(p^{\kappa}, t^n)$ can be computed from approximations
modulo~$(p^{\kappa + O(\log(n)^2)}, t^n)$ of~$g$ and~$h$.  A preliminary
version of the present work appeared in Vaccon's PhD thesis \parencite{Vac15}.
Naturally, the result also holds over unramified extensions of~$\Qp$,
see~\S\ref{sec:algo}.

\subsection{Applications}

\subsubsection{Newton sums}

The problem studied by \textcite{BosGonPer05} is the recovery of a polynomial given its Newton sums.
Let~$f\in\Zp[t]$ be a monic polynomial of degree~$d$, and let~$\nu_n$ be
the~$n${th} Newton sum of~$f$: if~$\alpha_1,\dotsc, \alpha_d$
are the roots of~$f$ in~$\overline{\Qp}$, then~$\nu_n$ is the
sum~$\alpha_1^n+\dotsb+\alpha_d^n$, it is an element of~$\Zp$.  How can we
recover~$f$ given Newton sums~$\nu_0,\dotsc, \nu_d$?  Let~$g$ be the polynomial~$x^d
f(1/x)$, and~$H_f$ be the generating function~$H_f(t) = \sum_{n\geq 0} \nu_{n+1}
t^n$.  Then~$g' = -H_f g$, so that~$g$ is a solution of a first-order linear
differential equation \parencite{Sch93}.  Therefore, knowing an approximation modulo~$(p^\kappa,
t^d)$ of~$H_f$ makes it possible to recover each coefficient of~$g$ (and hence~$f$) modulo~$p^{\kappa - \lfloor\log_p n\rfloor}$.

An interesting application is the computation over~$\Fp$ of composed products;
composed sums can be treated similarly~\parencite{BosGonPer05}.  Let~$f$ and~$g$ be monic polynomials
of~$\Fp[t]$ of degree~$d$ and~$e$ respectively, with associated
roots~$(\alpha_i)_{1\leq 1 \leq d}$ and~$(\beta_j)_{1\leq j\leq e}$
in~$\overline{\Fp}$.  We define the composed product of~$f$ and~$g$ to be
\[ f\otimes g = \prod_{i,j}(t - \alpha_i \beta_j ) = \mathop{res}_y(y^d f(t/y), g(y)), \]
this is a polynomial in~$\Fp[t]$ of degree~$de$.  Then~$H_{f\otimes g}$ is the
coefficient-wise product of the two power series~$H_f$ and~$H_g$ (also known as the Hadamard product).  This gives a
strategy to compute efficiently~$f\otimes g$.  Firsty, arbitrarily lift~$f$
and~$g$ as polynomials in~$\Zp[t]$, denoted~$\bar f$ and~$\bar g$. The composed
product~$\bar f \otimes \bar g$ is a polynomial in~$\Zp[t]$ and equals~$f\otimes g$
modulo~$p$.  Secondly, compute approximations
modulo~$(p^\kappa, t^{de})$ of~$H_{\bar f}$ and~$H_{\bar g}$, with~$\kappa = 1 + \lfloor \log_p(de) \rfloor$, and, with a
coefficient-wise product, an approximation modulo~$( p^\kappa, \smash{t^{de}})$
of~$H_{\bar f \otimes \bar g}$.  Thirdly, compute an approximation
modulo~$(p, t^{de+1})$ of~$\bar f \otimes \bar g$ using
Theorem~\ref{MainTheorem}, and deduce the value of~$f\otimes g$.

Another application of this procedure to root finding in finite fields is developed by \textcite{Grenet:2015}.
Our work does not improve the complexity given by \textcite{BosGonPer05}, in the case of a linear equation, but it
gives a simpler proof that generalizes to nonlinear differential equations.

\subsubsection{Isogeny computation}
\label{sec:intro-isogeny}

To compute normalized isogenies between elliptic curves, \textcite{Bostan:2008} and \textcite{LerSir08}
 studied the differential equation
\begin{equation}
\label{DiffEquQuad}y'^2 =g \cdot h(y),
\end{equation} where~$g$ and~$h$
are series in~$\Zpt$. In their context, this differential equation is known to admit a solution
in~$\Zpt$. Like the previous example, it comes from a lift of a problem over~$\Fp$.
This equation rewrites equivalently as~$y' = \smash{\sqrt{g}
\sqrt{h(y)}},$ and, when~$p\neq 2$, the series~$\sqrt{g}$ and~$\sqrt{h}$ are
still in~$\Zpt$, so we can apply Theorem~\ref{MainTheorem}.
The study of this equation when~$p=2$ is still an open problem.

In order to compute
an approximation of~$y$ modulo~$(p, t^{n+1})$,
we obtain that
it is enough to have approximations of~$g$ and~$h$ modulo~$(p^{1+ \lfloor\log_p n\rfloor}, t^n)$.
This improves upon the result of \textcite{LerSir08} which requires approximations modulo~$( p^{O(\log(n)^2)}, t^n)$.

\subsection*{Acknowledgements}
We are grateful to Alin Bostan, Xavier Caruso, 
Luca De~Feo, Reynald Lercier, Éric Schost and Kazuhiro Yokoyama 
for fruitful discussions.

\section{The algorithm}
\label{sec:algo}

We may consider the more general setting of an unramified finite
extension~$K$ of~$\Qp$. This
is useful, for example, for the computation of isogenies.
Let~$\cO_K$ denote the ring of integers of~$K$.
For example, we may naturally consider $K = \Qp$ and~$\cO_K = \Z_p$.

Let~$h \in \Opt$, a power series with integer coefficients, with~$h(0)=1$. 
For~$g\in\Kpt$, let~$Y(g)$ be the unique~$y\in \Kpt$ such that~$y(0)=0$
and~$y'=g\cdot h(y)$.  Existence and uniqueness are clear because the
differential equation rewrites equivalently into a well posed recurrence relation on the
coefficients of~$y$.

For~$g\in\Kpt$, let~$N_{g}$ denote the Newton operator:
\[ N_{g}(u) = u - h(u)\int \left(\frac{u'}{h(u)} - g \right), \]
where~$\int f$, for~$f\in\Kpt$, denotes the unique power series $F\in\Kpt$ such that~$F'=f$
and~$F(0)=0$.  

\begin{prop}
  \label{prop:newton-exact}
  Let~$g, u\in \Kpt$ be formal power series, and let~$n > 0$.
  If~$u = Y(g) \pmod{t^{n}}$ then
  \[ N_{g}(u) = Y(g) \pmod{t^{2n}}. \]
\end{prop}

\begin{proof} 
  Let~$e = {u'}/{h(u)} - g$.
  Since~$u = Y(g) \pmod{t^{n}}$,
  we have~$e = 0 \pmod{t^{n-1}}$ and~$\int {e} = 0 \pmod{t^n}$.
  With~$v = N_{g}(u)$, we compute
  \begin{align*}
    v' - g\, h(v) &= u' - u' h'(u) \smallint {e} - h(u) e - g\, h(v).
  \end{align*}
  Then, the first-order expansion of~$h(v)$ at~$h(u)$ gives
  \[ h(v) = h(u) - h'(u) h(u) \smallint {e} \mod{t^{2n}}, \]
  and, using the equality~$u' = h(u) e + g\, h(u)$,
  we obtain
  \[ v' - g\, h(v) = - e \, h(u) h'(u) \smallint {e} = 0 \mod{t^{2n - 1}}. \]
  This implies that~$v = Y(g) \pmod{t^{2n}}$.
\end{proof}

The iteration of the Newton operator leads to
Algorithm~\ref{algo:newton}.  In an exact setting, the correctness of
this procedure would be clear, thanks to Proposition~\ref{prop:newton-exact}.
In a \padic setting, where the coefficients of the power series~$g$ and~$h$ are
known with finite precision only, what can be obtained with Newton iteration is
not clear because the operation~$\smallint$ involves divisions.

\begin{algo}[t]
  \caption[The Newton iteration to solve a first-order differential equation.]
  {The Newton iteration to solve a first-order differential equation.
    \begin{description}
      \item[Input.] $g$ and~$h \in \Opt$ given modulo~$(p^\lambda, t^{n})$
        ---~that is, given as polynomials of degree less than~$n$ with
        coefficients in~$\cO_K/p^\lambda \cO_K$.
      \item[Output.] A power series~$u \in \Opt$ given modulo~$(p^\lambda, t^{n+1})$. 
      \item[Specification.]
        If~$Y(g) \pmod{t^{n+1}}$ has integer coefficients and if~$\lambda \geq \kappa + \lfloor \log_p n \rfloor$, then~$u$
        is an approximation of~$Y(g)$ modulo~$(p^{\kappa}, t^{n+1})$.
    \end{description}
  }
  \label{algo:newton}
  \begin{algorithmic}
    \Function{DSol}{$g$, $h$, $n$}
      \If{$n=0$}
      \State \textbf{return} $0 \pmod{t}$
      \Else
      \State $u \gets \textsc{DSol}(g, h, \lceil \frac{n-1}{2} \rceil)$
      \State \textbf{return} $N_g(u) \pmod{t^{n+1}}$
        \State { }\Comment \parbox[t]{.55\linewidth}{Compute at fixed precision~$\lambda$.}
      \EndIf
    \EndFunction
  \end{algorithmic}
\end{algo}

Let us begin with a quick analysis of Algorithm~\ref{algo:newton}.
On input~$(g,h,n)$, it performs a recursive call and
computes~$u = \textsc{DSol}(g, h, m)$, with~$m=\lceil \frac{n-1}{2} \rceil$.
Let us assume that~$u$ is an approximation modulo~$(p^\kappa, t^{m+1})$
of~$Y(g)$, for some~$\kappa > 0$.  Then~$N_g(u)$ is an approximation
modulo~$(p^{\kappa - \lfloor \log_p n \rfloor}, t^{n+1})$: Indeed, the
computation of~$N_g(u)$ involves divisions by the integers from~$2$ to~$n$ on distinct coefficients, so
the loss of precision is at most the maximum valuation of these integers, which
is~$\lfloor \log_p n \rfloor$.
Thus, if we define~$\mu(0)=0$ and~$\mu(n) = \lfloor \log_p n \rfloor + \mu(\lceil \frac{n-1}{2} \rceil)$,
we obtain that~$\textsc{DSol}(g, h, n)$ is an approximation modulo~$(p^{\lambda - \mu(n)}, t^{n+1})$
of~$Y(g)$.
We can check that~$\mu(n) = O(\log(n)^2)$.
Theorem~\ref{thm:correctness} improves on that analysis and shows that the precision of the result is at least~$\lambda - \lfloor \log_p n \rfloor$ and matches the intrinsic loss of precision.

We assume the \emph{fixed precision model} for computing with \padic numbers:
at precision~$\lambda$, it amounts to work over the ring~$\cO_K/p^\lambda\cO_K$.  When a
division~$a/b$ arises, with the approximation of~$a$ and~$b \in \cO_K$ given
in~$\cO_K/p^\lambda\cO_K$, three cases may arise:
\begin{itemize}
  \item $v_p(b) = 0$, in which case~$b$ is invertible in~$\cO_K/p^\lambda\cO_K$ and the division is well defined.
  \item $v_p(a) \geq v_p(b) > 0$, in which case~$a/b$ is in~$\cO_K$ but its approximation in~$\cO_K/p^\lambda\cO_K$ is not fully determined by the approximations of~$a$ and~$b$, so $a/b$ is arbitrarily defined in this model to be the class in~$\cO_K/p^\lambda\cO_K$ of the smallest integer~$c \geq 0$ such that~$a=bc \pmod{p^\lambda}$.%
  \item $v_p(a) < v_p(b)$, in which case~$a/b$ is not an integer and an error is raised.
\end{itemize}

The main argument for the correctness of Algorithm~\ref{algo:newton} is the
following proposition, proved in Section~\ref{sec:etude-diff}.

\begin{prop}
  \label{prop:Ypertub}
  Let~$n > 0$ and~$\kappa > 0$ (or~$\kappa > 1$ if~$p=2$) be integers, and let
  $g\in \Opt$ such that~$Y(g) \pmod{t^{n+1}}$ has integer coefficients.
  For any~$y\in\Kpt$ the following are equivalent:
  \begin{enumerate}
    \item $y = Y(\bar g) \pmod{t^{n+1}}$ for some power series~$\bar g \in \Opt$ such that~$\int(\bar g - g) = 0 \pmod{p^\kappa}$;
    \item $y = Y(g) \pmod{p^\kappa, t^{n+1}}$.
  \end{enumerate}
\end{prop}

\begin{theo}\label{thm:correctness}
  Algorithm~\ref{algo:newton} is correct:
  if~$\kappa > 0$ (or~$\kappa > 1$ if~$p=2$) and $\lambda \geq \lfloor \log_p n \rfloor + \kappa $, then for all~$g\in \Opt$ such that~$Y(g)$ has integer coefficients,
  the output of the procedure~$\textsc{DSol}(g,h,n)$ equals~$Y(g) \pmod{p^\kappa, t^{n+1}}$.

  Moreover, it performs~$\cO\left(M_{\cO_K}(p^\lambda,n) + C_h(p^\lambda, n) \right)$ bit operations,
  where~$M_{\cO_K}(p^\lambda,n)$ is the cost of computing the product of two polynomials of
degree~$n$ with coefficients in~$\cO_K/p^{\lambda}\cO_K$.
\end{theo}

\begin{proof}
  We proceed by induction on~$n$. The case~$n=0$ is trivial, so let us assume that~$n>0$.  Let~$g \in \Opt$ such that~$Y(g)$ has integer coefficients,
  let~$m = \lceil \frac{n-1}{2} \rceil$ and let~$u\in\Kpt$ be the output of~$\textsc{DSol}(g, h,
  m)$.  By induction hypothesis, $u = Y(g)
  \pmod{p^\kappa, t^{m+1}}$. (In particular~$u$ has integer coefficients.)
  By Proposition~\ref{prop:Ypertub}, this implies that~$y = Y(\bar g) \pmod{t^{m+1}}$
  for some~$\bar g\in \Opt$ such that~$\int(g - \bar g) \pmod{p^\kappa}$.
  Proposition~\ref{prop:newton-exact} gives that $Y(\bar g) = N_{\bar g}(u) \pmod{t^{n+1}}$
  and Proposition~\ref{prop:Ypertub} gives further that $Y(\bar g) = Y(g) \pmod{p^\kappa, t^{n+1}}$.
  We check that
  $N_g(u) = N_{\bar g}(u) - h(u) \smallint(g - \bar g)$
  and since~$h(u)$ has integer coefficients, this implies that
  $N_g(u) = Y(g) \pmod{p^\kappa, t^{n+1}}$.
  
  We now relate~$N_g(u)$ to the output of the procedure $\textsc{DSol}(g, h, n)$.
  Let~$e = u'/h(u) - g$.
  By definition, the output is~$N_{g}(u)  = u - h(u) \int e$, computed
  over~$\cO_K/p^\lambda\cO_K$, in the fixed precision model.  Let~$E$
  be the primitive~$\int e \pmod{t^{n+1}}$ computed in this model, so that the
  output is exactly~$u - h(u)E$.  Clearly~$E = \int e + \int \eta
  \pmod{t^{n+1}}$ for some~$\eta = 0 \pmod{p^\lambda}$ that reflects the
  indeterminacies in the divisions.
  Since~$\lambda \geq \lfloor \log_p n \rfloor + \kappa$, $\int\eta = 0 \pmod{p^\kappa}$ and
  thus, the output~$N_g(u) + \int \eta$ equals~$Y(g) \pmod{p^\kappa, t^{n+1}}$.
  This concludes the proof of correctness.

  Concerning the complexity, the last iteration involves a composition by~$h$
  with cost~$C_h(p^\lambda, n)$, a few multiplications with
  cost~$M_{\cO_K}(p^\lambda,n)$ and an inversion~$1/h(u)$ with cost
  $\cO\left(M_{\cO_K}(p^\lambda,n)\right)$ too with a Newton iteration
  \parencite{Kun74}.  With the assumption that the cost of an iteration is
  greater than twice the cost of the previous one, it is well known that the cost
  of a Newton algorithm is dominated by the cost of the last iteration, which gives the result.
\end{proof}

The condition~$\lambda \geq \lfloor \log_p n \rfloor + \kappa$ cannot be improved further: it matches the \emph{intrinsic} loss of precision. This is shown, for example, by the differential equation~$y' = \smash{at^{b-1}}$, with~$a\in K$, whose solution is~$\frac ab t^b$. If we take~$b = p^{\lfloor \log_p n \rfloor}$ and if~$a$ is known at precision~$\lambda$ then~$y$ is known at precision no more than~$\lambda - \lfloor \log_p n \rfloor$.

\section{Differential precision}
\label{sec:etude-diff}

We apply the method of \textcite{CarRoeVac14} to study the loss in precision in
the resolution of the differential equation \eqref{MainDiffEqu} and give a
proof of Proposition~\ref{prop:Ypertub}.

Let~$n>0$ and
let~$E$ and~$F$ denote respectively the two $n$-dimensional $K$-vector spaces $\Kpt/(t^{n})$ and~$t\Kpt/(t^{n+1})$.
Let~$\cY$ be the polynomial map
\begin{align*}
  \cY : E &\longrightarrow F \\
        [ u ] &\longmapsto [ Y(u) ],
\end{align*}
which is well defined because the~$n+1$ first coefficients of~$Y(u)$ depend only on the~$n$ first coefficients of~$u$.
Let~$g\in\Opt$ be such that~$\cY(g)$ has integer coefficients in the monomial basis.
Let~$\ud\cY$ denote the first differential of~$\cY$: for any~$g\in E$, $\ud_g \cY$ is a linear map~$E\to F$.
Let~$d^k\cY$ denote the higher differentials: for any~$g\in E$, $\ud^k_g \cY$ is a multilinear map~$E^k\to F$.

\begin{lem}
  \label{lem:diffY}
  For any~$w\in E$, $\ud \cY(w) = h(\cY) \smallint w$. 
  Moreover, for any~$k\geq 1$, there exists a polynomial~$P_k\in \Z[u_0,\dotsc,u_{k-1}]$
  such that for any~$w_1,\dotsc,w_k\in E$,
  \begin{multline}\label{eqn:higher-diff}
  \ud^k \cY(w_1,\dotsc,w_k) = \\ P_k\left( h(\cY), h'(\cY), \dotsc, h^{(k-1)}(\cY) \right) \prod_{i=1}^k \smallint w_i.
  \end{multline}
\end{lem}

\begin{proof}
  Differentiating with respect to~$g$ the defining relation~$\cY(g)' = g \cdot h(\cY(g)) \pmod{t^{n}}$
  leads to
  \[ \left[ \ud_g \cY(w) \right]' = w \cdot h(\cY(g)) + g \cdot h'(\cY(g)) \cdot \ud_g \cY(w), \]
  which is a first-order inhomogeneous linear differential equation in~$\ud_g \cY(w)$.
  The initial condition~$\ud_g \cY(w)(0) = 0$ determines a unique solution, namely $h(\cY(g)) \int w$.

  The second claim follows by induction. Equation~\eqref{eqn:higher-diff} holds for~$k=1$ with~$P_1=u_0$;
  differentiating it leads to the  recurrence relation
  \[ P_{k+1}(u_0,\dotsc,u_k) = u_0 \sum_{i=0}^{k-1} \frac{\partial P}{\partial u_i} u_{i+1}. \qedhere \]
\end{proof}

The space~$F$ is endowed with the maximum norm in the
monomial basis, denoted by~$\|\cdot\|_F$. In particular, an element~$u$ of~$F$
has integer coefficients if and only if~$\|u\|_F \leq 1$.
The space~$E$ is endowed with the norm
\[ \|u\|_E \eqdef \left\| \smallint u \right \|_F. \]
Let~$\| \ud^k_g \cY \|$ denote the operator norm of~$\ud_g^k \cY$, that is
\begin{equation}
 \| \ud^k_g \cY \| = \sup\left\{ \| \ud_g^k\cY(w_1,\dotsc,w_k) \|_F \st  \|\smallint w_i\|_F \leq 1 \right\}.  
  \label{eqn:opnorm}
\end{equation}

\begin{lem}
  \label{lem:normDY}
  $\| \ud^k_g \cY \| \leq 1$, for any~$k\geq 1$.
\end{lem}

\begin{proof}
  Since~$h$ and~$\cY(g)$ have integer coefficients, this follows easily from Lemma~\ref{lem:diffY}
  and Equation~\eqref{eqn:opnorm}.
\end{proof}

\begin{prop}\label{prop:CRV}
  For any~$\varepsilon \leq \frac1p$ (or $\varepsilon \leq \frac14$ for~$p=2$),
  \[ \cY(g + B_\varepsilon) = \cY(g) + \ud_g\cY(B_\varepsilon), \]
  where~$B_\varepsilon = \left\{ w\in E \ :\  \|w\|_E \leq \varepsilon \right\}$.
\end{prop}

\begin{proof}
  We apply the result of \textcite[Corollary~3.16]{CarRoeVac14}.  Using their
  notations, we can use $C=1$ because the closed ball of radius~$1$ in~$F$
  (that is the set of elements with integer coefficients) is included
  in~$\ud_g\cY(B_1)$: Indeed for any~$u\in F$ we have~$\ud_g \cY\left( ( u / \cY(g) )' \right) = u$,
  and if~$\|u\|_F \leq 1$ then
  $ \| ( u / \cY(g) )' \|_E = \| ( u / \cY(g) ) \|_F \leq 1$,
  because~$u$ and~$\cY(g)$ have integer coefficients.

  For~$k\geq 2$, let~$M_k$ denote~$\| \frac{1}{k!} \ud^k_g \cY \|$.
  By Lemma~\ref{lem:normDY}, this is simply~$| \frac{1}{k!}|$.
  Corollary~3.16 (ibid.),  with~$\rho=1$ in their notations,  implies that
  $\cY(g + B_\varepsilon) = \cY(g) + \ud_g\cY(B_\varepsilon)$ as long as~$\varepsilon$ satisfies
  \[ \varepsilon < \exp\left(\inf_{k\geq 2} \frac{-\log M_k}{k-1}  \right) = \inf_{k\geq 2} p^{-\frac{v_p(k!)}{k-1}}. \]
  Let~$A$ denote the right-hand side. %
  Legendre's formula for the \padic valuation of~$k!$ shows that
  $v_p(k!) \leq \frac{k}{p-1}$. Therefore~$A \geq p^{-2/(p-1)}$.
  For~$p\geq 5$, this bound gives~$A > \frac{1}{p}$, which proves the claim.
  For~$p = 3$, we have
  \[ A \geq \min\left( 3^{-v_3(2!)}, \inf_{k\geq 3} 3^{-\frac{v_3(k!)}{k-1}} \right) \geq \min\left(1, 3^{-\frac34}\right) > 3^{-1}, \]
  and for~$p=2$, we have
  \[ A \geq \min\left( 2^{-v_2(2!)}, \inf_{k\geq 3} 2^{-\frac{v_2(k!)}{k-1}} \right) \geq \min\left(2^{-1}, 2^{-\frac32}\right) > 2^{-2}, \]
  which concludes the proof.
\end{proof}

\begin{proof}
  [Proof of Proposition~\ref{prop:Ypertub}]
  Let~$\varepsilon = p^{-\kappa}$.
  The norm of an element~$v\in F$ is given by~$p^{-\lambda}$
  where~$\lambda$ is the largest integer such that~$v = 0\pmod{p^\lambda}$.
  Since~$h(\cY(g))$ is invertible modulo~$(p^{\lambda}, t^{n+1})$, for any~$\lambda>0$,
  this shows that~$\|v\|_F = \|h(\cY(g)) v\|_F$ for any~$v\in F$.
  Therefore, with Lemma~\ref{lem:diffY} and the definition of the norms,
  \[ \ud_g\cY(B_\varepsilon) = \left\{ v \in F \st v = 0 \mod p^\kappa \right\}. \]
  Moreover, $B_\varepsilon = \left\{ u \in E \st \smallint u = 0 \mod{p^\kappa} \right\}$,
  and Proposition~\ref{prop:Ypertub} now appears as a rewording of Proposition~\ref{prop:CRV}.
\end{proof}

\section{Experiments}
\label{sec:implementation-equa-diff}

\begin{figure}[t]
  \begin{center}
    \includegraphics[width=\linewidth]{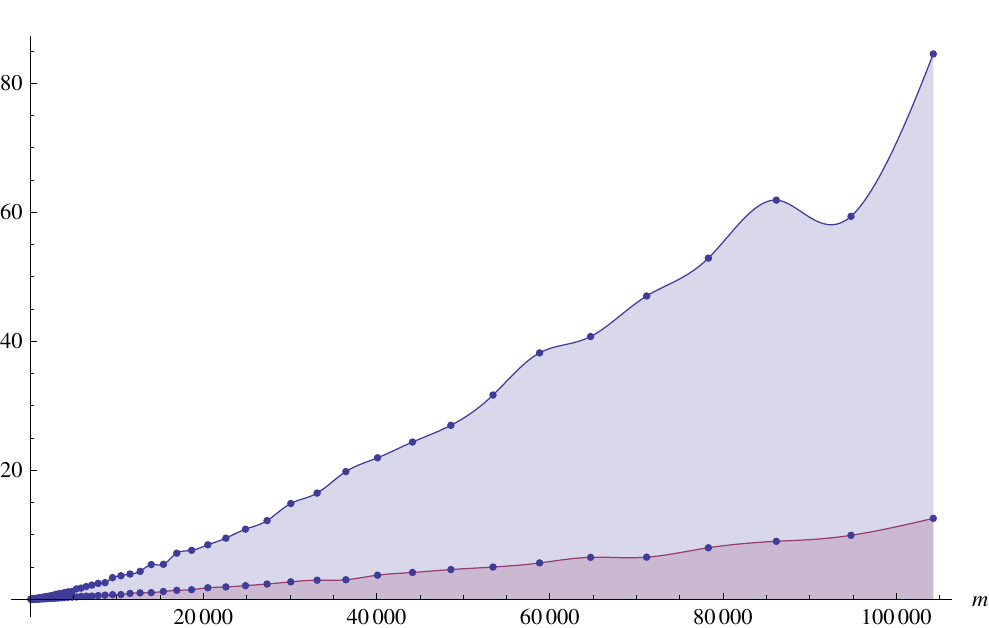}
  \end{center}
  \caption{Timings in seconds, measured on a laptop, of Algorithm~\ref{algo:newton} run at precision~$\lambda_\text{old}$ (upper curve) and~$\lambda_\text{new}$ (lower curve) in order to compute an approximation modulo~$(5,t^{4m+1})$ of the solution of Equation~\eqref{eqn:deqiso}.}
  \label{fig:timings}
\end{figure}

\begin{figure}[t]
  \begin{center}
    \includegraphics[width=\linewidth]{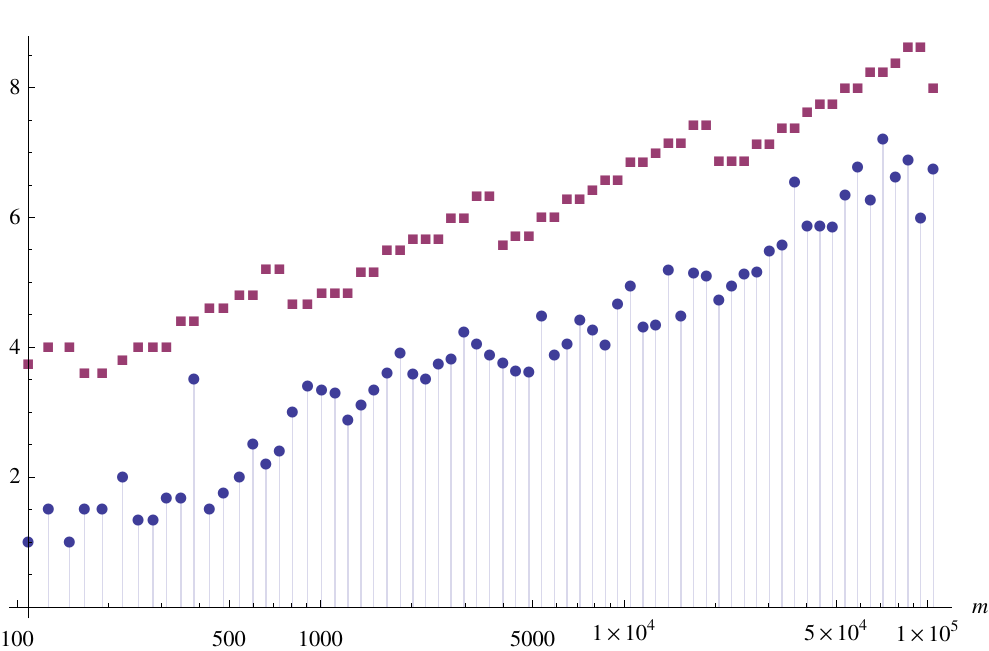}
  \end{center}
  \caption{ {Practical speedup obtained with the new precision analysis compared with the theoretical improvement ($m$-axis in logarithmic scale).\\
    ({\tiny $\blacksquare$}) ratio~$\lambda_\text{old}/\lambda_\text{new}$; ($\bullet$) actual speedup.
  } }
  \label{fig:speedup}
\end{figure}

Let us consider the differential equation
\begin{equation}
 y' = \sqrt{\frac{1 + \frac{1}{4} m^2 y^2 + m^6y^6}{1+\frac{1}{4}t^2 + t^6}}, \quad y(0)=0,  
  \label{eqn:deqiso}
\end{equation}
inspired from algorithms for computing isogenies
\parencite{Bostan:2008,LerSir08}.  Using an implementation in Magma \parencite{magma} of
Algorithm~\ref{algo:newton}, we computed the power series expansion of~$y
\pmod{5, t^{4m+1}}$ for several~$m$.  We compared (Figure~\ref{fig:timings}) the CPU time spent on the computation when using on
the one hand the precision~$\lambda_\text{new} = 1+\lfloor \log_5(4 m) \rfloor$,
following Theorem~\ref{thm:correctness}, and using on the other hand the
precision~$\lambda_\text{old} = 1+\mu(4m) = \cO( \log(m)^2 )$ found by a straightforward precision analysis --- see the discussion in~\S\ref{sec:algo} for the definition of~$\mu$.
For example, with~$m=104281$, we compute~$\lambda_\text{old}=72$ and~$\lambda_\text{new}=9$.
The number of arithmetic operations performed does not depend on the precision~$\lambda$, only on~$m$, but the number of bit operations does since the base ring for the computation is~$\Z/5^\lambda\Z$.
Thus, the expected speedup is~$\lambda_{\text{old}}/\lambda_\text{new}$, which is close to what we observed (Figure~\ref{fig:speedup}).
The implementation is available at
{\center\url{https://gist.github.com/lairez/d648b0d7b5392d0fef74}.

}

\printbibliography

\end{document}